\documentclass[draftclsnofoot, onecolumn, 11pt]{IEEEtran}

\usepackage[T1]{fontenc}
\usepackage[utf8]{inputenc}
\usepackage{amsmath,amssymb,amsfonts,mathrsfs,bm}
\usepackage{amsthm}
\usepackage{dsfont}
\usepackage{cite}
\usepackage[shortlabels]{enumitem}
\usepackage{breqn}
\usepackage[pdftex]{graphicx}
\usepackage{subfig}
\usepackage{float}
\graphicspath{{./Figures/}}
\usepackage{algorithm,algorithmic}
\usepackage{multirow}
\usepackage{color}
\usepackage[table]{xcolor}
\usepackage[normalem]{ulem}
\usepackage{array}
\newcolumntype{L}[1]{>{\raggedright\let\newline\\\arraybackslash\hspace{0pt}}m{#1}}
\newcolumntype{C}[1]{>{\centering\let\newline\\\arraybackslash\hspace{0pt}}m{#1}}
\newcolumntype{R}[1]{>{\raggedleft\let\newline\\\arraybackslash\hspace{0pt}}m{#1}}
\usepackage{url}
\usepackage[hidelinks]{hyperref}

\ifx\useTheoremCounter\undefined
\newtheorem{Theorem}{Theorem}
\newtheorem{Corollary}{Corollary}
\newtheorem{Proposition}{Proposition}
\newtheorem{Lemma}{Lemma}
\else
\newtheorem{Theorem}{Theorem}

\newtheorem{Lemma}[Theorem]{Lemma}
\fi

\newtheorem{Remark}{Remark}
\newtheorem{Assumption}{Assumption}

\theoremstyle{remark}

\newcommand{\Real}{\mathbb{R}}
\newcommand{\Nat}{\mathbb{N}}

\newcommand{\calA}{\mathcal{A}}
\newcommand{\calB}{\mathcal{B}}
\newcommand{\calC}{\mathcal{C}}
\newcommand{\calD}{\mathcal{D}}
\newcommand{\calE}{\mathcal{E}}
\newcommand{\calF}{\mathcal{F}}
\newcommand{\calG}{\mathcal{G}}

\newcommand{\calM}{\mathcal{M}}
\newcommand{\calN}{\mathcal{N}}

\newcommand{\calR}{\mathcal{R}}
\newcommand{\calS}{\mathcal{S}}

\newcommand{\lk}{\ell k}

\newcommand{\bmd}{\bm{d}}

\newcommand{\bms}{\bm{s}}

\newcommand{\bmu}{\bm{u}}

\newcommand{\bmv}{\bm{v}}

\newcommand{\bmw}{\bm{w}}

\newcommand{\bgamma}{\bm{\gamma}}

\newcommand{\bepsilon}{\bm{\epsilon}}

\newcommand{\bpsi}{\bm{\psi}}

\DeclareMathOperator{\Tr}{Tr}
\DeclareMathOperator{\Vc}{vec}

\newcommand{\E}{\mathbb{E}}
\newcommand{\T}{{\sf T}}

\newcommand{\norm}[1]{{\left\lVert{#1}\right\rVert}}

\newcommand{\err}[1]{\widetilde{#1}}

\newcommand{\diag}[1]{\operatorname{diag}\left\{{#1}\right\}}
\newcommand{\col}[1]{\operatorname{col}\left\{{#1}\right\}}

\newcommand{\msout}[1]{\text{\color{green} \sout{\ensuremath{#1}}}}
\newcommand{\del}[1]{{\color{green}\ifmmode \msout{#1}\else\sout{#1}\fi}}

\definecolor{mypurple}{rgb}{0.49, 0.18, 0.56}
\definecolor{myorange}{rgb}{0.99, 0.35, 0.11}
\definecolor{mygreen}{rgb}{0.18, 0.55, 0.34}
\newcommand{\wydel}[1]{{\color{mygreen}\ifmmode \msout{#1}\else\sout{#1}\fi}}

\begin{document}
\title{An Event-based Diffusion LMS Strategy}

\author{Yuan Wang, Wee~Peng~Tay, and Wuhua Hu
	\thanks{Y. Wang and W. P. Tay are with the School of Electrical and Electronic Engineering, Nanyang Technological University, Singapore. Emails: ywang037@e.ntu.edu.sg, wptay@ntu.edu.sg.} 
	\thanks{W. Hu was with the School of Electrical and Electronic Engineering, Nanyang Technological University, Singapore, and now is with the Department of Artificial Intelligence for Applications, SF Technology Co. Ltd, Shenzhen, China. E-mail: wuhuahu@sf-express.com.}
}

\maketitle

\begin{abstract}
We consider a wireless sensor network consists of cooperative nodes, each of them keep adapting to streaming data to perform a least-mean-squares estimation, and also maintain information exchange among neighboring nodes in order to improve performance. For the sake of reducing communication overhead, prolonging batter life while preserving the benefits of diffusion cooperation, we propose an energy-efficient diffusion strategy that adopts an event-based communication mechanism, which allow nodes to cooperate with neighbors only when necessary. We also study the performance of the proposed algorithm, and show that its network mean error and MSD are bounded in steady state. Numerical results demonstrate that the proposed method can effectively reduce the network energy consumption without sacrificing steady-state network MSD performance significantly.
\end{abstract}

\section{Introduction}
In the era of big data and Internet-of-Things (IoT), ubiquitous smart devices continuously sense the environment and generate large amount of data rapidly. To better address the real-time challenges arising from online inference, optimization and learning, distributed adaptation algorithms have become especially promising and popular compared with traditional centralized solutions. As computation and data storage resources are distributed to every sensor node in the network, information can be processed and fused through local cooperation among neighboring nodes, and thus reducing system latency and improving robustness and scalability. Among various implementations of distributed adaptation solutions \cite{Bertsekas:J97,RabbatNowak:J05,Bogdanovic:J14,XiaoBoydLall:C06,Nedic:J10,SrivaNedic:J11}, diffusion strategies are particularly advantageous for continuous adaptation using constant step-sizes, thanks to their low complexity, better mean-square deviation (MSD) performance and stability \cite{CatSayed:J10,ZhaoSayed:J12,TuSayed:J12,SayedTuChen:M13,Sayed:BC14,Sayed:IP14}. Therefore diffusion strategies have attracted a lot of research interest in recent years for both single-task scenarios where nodes share a common parameter of interest \cite{HuTay:J15,Chencheng:J15,AbdoleeChampagne:J16,RadLabeau:J16,PiggottSolo:J16,NtemosPlataChaves:J17,Chengcheng:J18}, and multi-task networks where parameters of interest differ among nodes or groups of nodes \cite{PlataChaves:J15,NassifRichard:J16,ChenRichard:J17,WangTayHu:J17,FernandezGarcia:J17}. 

In diffusion strategies, each sensor communicates local information to their neighboring sensors in each iteration. However, in IoT networks, devices or nodes usually have limited energy budget and communication bandwidth, which prevent them from frequently exchanging information with neighboring sensors. Several methods to improve energy efficiency in diffusion have been proposed in the literature, and these can be divided into two main categories: reducing the number of neighbors to cooperate with \cite{ZhaoSayed:C12,Arablouei:J15,HuangYangShen:J17}; and reducing the dimension of the local information to be transmitted \cite{ArabloueiLMS:J14,SayinKozat:J14,HarraneRicard:C16}. These methods either rely on additional optimization procedures, or use auxiliary selection or projection matrices, which require more computation resources to implement.

Unlike time-driven communication where nodes exchange information at every iteration, event-based communication mechanisms allow nodes only trigger communication with neighbors upon occurrence of certain meaningful events. This can significantly reduce energy consumption by avoiding unnecessary information exchange especially when the system has reached steady-state. It also allows every node in the network to share the limited bandwidth resource so that channel efficiency is improved. Such mechanisms have been developed for state estimation, filtering, and distributed control over wireless sensor networks\cite{WuJiaJohansson:J13,HanMoWu:J15,LiuWangHe:J15,Mohammadi:J17,Seyboth:J13,GarciaCao:J2014,HuLiuFeng:J16,XingWenGuo:J17}, but have not been fully investigated in the context of diffusion adaptation. In \cite{UtluKilicKozat:J17}, the author proposes a diffusion strategy where every entry of the local intermediate estimates are quantized into values of multiple levels before being transmitted to neighbors, communication is triggered once quantized local information goes through a quantization level crossing. The performance of this method relies largely on the precision of selected quantization scheme. However, choosing a suitable quantization scheme with desired precision, and requiring every node being aware of same quantization scheme is practically difficult for online adaptation where parameter of interest and environment may change over time.

In this paper, we propose an event-based diffusion strategy to reduce communication among neighboring nodes while preserve the advantages of diffusion strategies. Specifically, each node monitors the difference between the full vector of its current local update and the most recent intermediate estimate transmitted to its neighbors. A communication is triggered only if this difference is sufficiently large. We provide a sufficient condition for the mean error stability of our proposed strategy, and an upper bound of its steady-state network mean-squared deviation (MSD). Simulations demonstrate that our event-based strategy achieves a similar steady-state network MSD as the popular adapt-then-combine (ATC) diffusion strategy but a significantly lower communication rate.

The rest of this paper is organized as follows. In Section~\ref{sec:Preliminaries}, we introduce the network model, problem formulation and discuss prior works. In Section~\ref{sec:EB-ATC}, we describe our proposed event-based diffusion LMS  strategy and analyze its performance. Simulation results are demonstrated in Section~\ref{sec:simulation} followed by concluding remarks in Sections~\ref{sec:conclusion}.

\emph{Notations.} 
Throughout this paper, we use boldface characters for random variables, plain characters for realizations of the corresponding random variables as well as deterministic quantities. In addition, we use upper-case characters for matrices and lower-case ones for vectors and scalars. The notation $I_N$ is an $N\times N$ identity matrix. The matrix $A^T$ is the transpose of the matrix $A$, $\lambda_n(A)$, and $\lambda_{\min}(A)$ is the $n$-th eigenvalue and the smallest eigenvalue of the matrix $A$, respectively. Besides, $\rho(A)$ is the spectral radius of $A$. The operation $A\otimes B$ denotes the Kronecker product of the two matrices $A$ and $B$. The notation $\norm{\cdot}$ is the Euclidean norm, $\norm{\cdot}_{b,\infty}$ denotes the block maximum norm\cite{Sayed:BC14}, while $\norm{A}^2_\Sigma \triangleq A^*\Sigma A$. We use $\diag{\cdot}$ to denote a matrix whose main diagonal is given by its arguments, and $\col{\cdot}$ to denote a column vector formed by its arguments. The notation $\Vc(\cdot)$ represents a column vector consisting of the columns of its matrix argument stacked on top of each other. If $\sigma = \Vc(\Sigma)$, we let $\norm{\cdot}_\sigma = \norm{\cdot}_\Sigma$, and use either notations interchangeably. 

\section{Data Models and Preliminaries}\label{sec:Preliminaries}
In this section, we first present our network and data model assumptions. We then give a brief description of the ATC diffusion strategy.

\subsection{Network and Data Model}
Consider a network represented by an undirected graph $G=(V,E)$, where $V=\{1,2,\cdots,N\}$ denotes the set of nodes, and $E$ is the set of edges. Any two nodes are said to be connected if there is an edge between them. The neighborhood of each node $k$ is denoted by $\calN_k$ which consists of node $k$ and all the nodes connected with node $k$. Since the network is assumed to be undirected, if node $k$ is a neighbor of node $\ell$, then node $\ell$ is also a neighbor of node $k$. Without loss of generality, we assume that the network is connected.

Every node in the network aims to estimate an unknown parameter vector $w^\circ\in\mathbb{R}^{M\times 1}$. At each time instant $i\ge0$, each node $k$ observes data $\bmd_k(i)\in\Real$ and $\bmu_k(i)\in\Real^{M\times 1}$, which are related through the following linear regression model:
\begin{align}
\bmd_k(i) = \bmu^\T_k(i) w^\circ + \bmv_k(i), \label{eq:datamodel}
\end{align}
where $\bmv_k(i)$ is an additive observation noise. We make the following assumptions. 

\begin{Assumption}\label{asmp:regressor}
The regression process $\{\bmu_{k,i}\}$ is zero-mean, spatially independent and temporally white. The regressor $\bmu_k(i)$ has positive definite covariance matrix $R_{u,k}=\E\left[\bmu_k(i)\bmu^\T_k(i)\right]$.  
\end{Assumption}

\begin{Assumption}\label{asmp:noise}
The noise process $\{\bmv_k(i)\}$ is spatially independent and temporally white. The noise $\bmv_k(i)$ has variance $\sigma^2_{v,k}$, and is assumed to be independent of the regressors $\bmu_\ell(j)$ for all $\{k, \ell, i, j\}$.
\end{Assumption}

\subsection{ATC Diffusion Strategy}
To estimate the parameter $w^\circ$, the network solves the following least mean-squares (LMS) problem:
\begin{align}\label{network_obj}
\min_{w}\; \sum_{k=1}^{N} J_k(w),
\end{align}
where for each $k\in V$,
\begin{align}
J_{k}(w) =  \sum_{k\in \calN_k} \E \left| \bmd_k(i) - \bmu_k(i)^\T w \right|^2.
\end{align}
The ATC diffusion strategy \cite{CatSayed:J10,Sayed:BC14} is a distributed optimization procedure that attempts to solve \eqref{network_obj} iteratively by performing the following local updates at each node $k$ at each time instant $i$:
\begin{align}
\bpsi_k(i) &= \bmw_k(i-1) + \mu_k \bmu_k(i) \left(  \bmd_k(i) - \bmu_k(i)^\T \bmw_k(i-1)\right) ,\label{eq:ATC-A}\\
\bmw_{k,i} &= \sum \limits_{\ell \in \calN} a_{\ell k} \bpsi_{\ell,i}, \label{eq:ATC-C}
\end{align}
where $\mu_k > 0$ is a chosen step size. The procedure in \eqref{eq:ATC-A} is referred to as the \textit{adaptation} step and \eqref{eq:ATC-C} is the \textit{combination} step. The combination weights $\{ a_{\ell k} \}$ are non-negative scalars and satisfy:
\begin{align}
{a_{\ell k}} \ge 0,\;\; {\sum \limits_{\ell =1}^{N} a_{\ell k} = 1}, \;\; {a_{\ell k} = 0}, \;  \text{if} \; {\ell \notin \calN_k}{.} \label{eq:weights_A}
\end{align}
The local estimates $\bmw_{k,i}$ in the ATC strategy are shown to converge in mean to the true parameter $w^\circ$ if the step sizes $\mu_k$ are chosen to be below a particular threshold \cite{CatSayed:J10,Sayed:BC14}.

\section{Event-Based Diffusion}\label{sec:EB-ATC}

We consider a modification of the ATC strategy so that the local intermediate estimate $\bpsi_k(i)$ of each node $k$ is communicated to its neighbors only at certain trigger time instants $s_k^n$, $n=1, 2, \ldots$. Let $\overline{\bpsi}_k(i)$ be the last local intermediate estimate node $k$ transmitted to its neighbors at time instant $i$, i.e., 
\begin{align}
\overline{\bpsi}_k(j)=\bpsi_k(s^n_k), \text{ for } j\in\left[s^n_k, s^{n+1}_k \right) .
\end{align}
Let $\bepsilon^-_k(i)$ be the \emph{a prior} gap defined as
\begin{align}
\bepsilon^-_k(i) = \bpsi_k(i)-\overline{\bpsi}_k(i-1). \label{eq:gap_aprior}
\end{align}
Let $f\left(\bepsilon^-_k(i)\right) = \norm{\bepsilon^-_k(i)}^2_{Y_k}$, where $Y_k$ is a positive semi-definite weighting matrix. 

For each node $k$, transmission of its local intermediate estimate $\bpsi_k(i)$ is triggered whenever 
\begin{align}
f\left(\bepsilon^-_k(i)\right)>\delta_k(i)>0,
\end{align}
where $\delta_k(i)$ is the threshold adopted by node $k$ at time $i$. 

In this paper, we allow the thresholds to be time-varying. We further assume $\{\delta_k(i)\}$ of each node $k$ are upper bounded, and let
\begin{align}
\delta_k = \sup\{\delta_k(i)|i>0\} .
\end{align}
In addition, we define binary variables $\{\bgamma_k(i)\}$ such that $\bgamma_k(i)=1$ if node $k$ transmits at time instant $i$, and 0 otherwise. The sequence of triggering time instants $0\le s^1_k\le s^2_k\le \ldots$ can then be defined recursively as
\begin{align}
s^{n+1}_k = \min \{i\in\Nat|i>s^n_k,\bgamma_k(i)=1\}.
\end{align}

For every node in the network, we apply the event-based adapt-then-combine (EB-ATC) strategy detailed in Algorithm~\ref{al:trigrule}. Note that every node always combines its own intermediate estimate regardless of the triggering status. A succinct form of the EB-ATC can be summarized as the following equations,
\begin{align}
\bpsi_k(i) &= \bmw_k(i-1) + \mu_k \bmu_k(i) \left(  \bmd_k(i) - \bmu_k(i)^\T \bmw_k(i-1)\right), \label{eq:local_update}  \\
\bmw_k(i) &= a_{kk}\bpsi_k(i)+ \sum\limits_{\ell\in\calN_k\backslash k} a_\lk \overline{\bpsi}_\ell(i). \label{eq:eb_atc_combination}
\end{align}

\begin{algorithm}[!tb]
	\caption{Event-based ATC Diffusion Strategy (EB-ATC)}
	\label{al:trigrule}
	\begin{algorithmic}[1] 
		\FOR {every node $k$ at each time instant $i$}
		\STATE \textit{\textbf{Local Update:}}
		\STATE Obtain intermediate estimate $\bpsi_k(i)$ using \eqref{eq:ATC-A}
		\vspace{0.25cm}
		\STATE \textit{\textbf{Event-based Triggering:}}
		\STATE Compute $\bepsilon_k^-(i)$ and $f\left(\bepsilon^-_k(i)\right)$.
		\IF {$f\left(\bepsilon^-_k(i)\right)>\delta_k(i)$}
		\STATE (i) Trigger the communication, broadcast local update $\bpsi_{k,i}$ to every neighbors $\ell\in\calN_k$.
		\STATE (ii) Mark $\bgamma_k(i)=1$, and update $\overline{\bpsi}_\ell(i) = \bpsi_\ell(i)$.
		\ELSIF {$f\left(\bepsilon^-_k(i)\right)\le\delta_k(i)$} 
		\STATE (i) Keep silent. 
		\STATE (ii) Mark $\bgamma_k(i)=0$, and update $\overline{\bpsi}_\ell(i) = \overline{\bpsi}_\ell(i-1)$.
		\ENDIF
		\vspace{0.25cm}
		\STATE \textit{\textbf{Diffusion Combination}}
		\STATE $\bmw_k(i) = a_{kk}\bpsi_k(i)+ \sum\limits_{\ell\in\calN_k\backslash k} a_\lk \overline{\bpsi}_\ell(i)$
		\ENDFOR
	\end{algorithmic}	
\end{algorithm}

\section{Performance Analysis}
In this section, we study the mean and mean-square error behavior of the EB-ATC diffusion strategy. 

\subsection{Network Error Recursion Model}
In order to facilitate the analysis of error behavior, we first define some necessary symbols and derive the recursive equations of errors across the network. To begin with, the error vectors of each node $k$ at time instant $i$ are given by
\begin{align}
\err{\bpsi}_k(i) &= w^\circ - \bpsi_k(i),  \label{eq:error_phi_k}\\
\err{\bmw}_k(i) &= w^\circ - \bmw_k(i).  \label{eq:error_w_k}
\end{align}
Recall that under EB-ATC each node only combines the local updates $\{\overline{\bpsi}_\ell(i)|\ell\in\calN_k\}$ that were previously received from its neighbors. Therefore, we also introduce the \emph{a posterior} gap $\bepsilon_k(i)$ defined as: 
\begin{align}
\bepsilon_k(i) = \bpsi_k(i) - \overline{\bpsi}_k(i), \label{eq:gap_aposterior}
\end{align}
to capture the discrepancy between the local intermediate estimate $\bpsi_k(i)$ and the estimate $\overline{\bpsi}_k(i)$ that is available at neighboring nodes. We have
\begin{align}
\bepsilon_k(i)=
\begin{cases}
0, &\text{ if } \norm{\bepsilon^-_k(i)}_{Y_k}^2>\delta_k(i),\\
\bepsilon^-_k(i), &\text{ otherwise}.
\end{cases} \label{eq:aposteriorgap_cases}
\end{align}
From \eqref{eq:aposteriorgap_cases}, we have the following result.

\begin{Lemma}\label{lemma:lemma1}
	The \emph{a posterior} gap $\bepsilon_k(i)$ is bounded, and $\norm{\bepsilon_k(i)}\le\left(\frac{\delta_k}{\lambda_{\min}(Y_k)}\right)^{\frac{1}{2}}$.
\end{Lemma}
\begin{proof}
See Appendix~\ref{sec:appdx_A}
\end{proof}
Collecting the iterates  $\err{\bpsi}_{k,i}$, $\err{\bmw}_{k,i}$, and $\bepsilon_k(i)$ across all nodes we have,
\begin{align}
\err{\bpsi}(i) &= \col{\left(\err{\bpsi}_k(i)\right) _{k=1}^N}, \label{eq:error_phi}\\
\err{\bmw}(i) &= \col{\left(\err{\bmw}_k(i)\right)_{k=1}^N}, \label{eq:error_w} \\
\bepsilon(i) &= \col{\left(\bepsilon_k(i)\right)_{k=1}^N}. \label{eq:gap_ap_network}
\end{align}
Subtracting both sides of \eqref{eq:local_update} from $\bmw^\circ$, and applying the data model \eqref{eq:datamodel}, we obtain the following error recursion for each node $k$:
\begin{align}
\err{\bpsi}_k(i) = \left(I_M-\mu_k\bmu_k(i)\bmu_k^\T(i) \right) \err{\bmw}_k(i) - \mu_k \bmu_k(i)\bmv_k(i). \label{eq:errorpsi}
\end{align}
Note that by resorting to \eqref{eq:gap_aposterior}, the local combination step \eqref{eq:eb_atc_combination} can be expressed as
\begin{align}
\bmw_k(i) = a_{kk}\bpsi_k(i)+\sum \limits_{\ell \in \calN_k\backslash k} a_\lk \left( \bpsi_\ell(i) - \bepsilon_\ell(i) \right), 
\end{align}
then subtract both sides of the above equation from $w^\circ$ we obtain
\begin{align}
\err{\bmw}_k(i)=\sum_{\ell\in\calN_k} a_\lk \err{\bpsi}_\ell(i) + \sum_{\ell\in\calN_k\backslash k}a_\lk\bepsilon_\ell(i) . \label{eq:errorw}
\end{align}
Let $A$ be the matrix whose $(\ell,k)$-th entry is the weight $a_\lk$, also we introduce matrix $C=A-\diag{(a_{kk})_{k=1}^N}$. Then relating \eqref{eq:error_w}, \eqref{eq:gap_ap_network}, \eqref{eq:errorpsi}, and \eqref{eq:errorw} yields the following recursion:
\begin{align}
\err{\bmw}(i) = \bm{\calB}(i) \err{\bmw}(i-1) - \calA^\T\calM \bms(i) + \calC^\T\bepsilon(i), \label{eq:error_network}
\end{align}
where
\begin{align}
\calA &= A \otimes I_M,\; \calC = C \otimes I_M\\
\bm{\calB}(i) &= \calA^\T\left(I_{MN}-\calM \bm{\calR}_u(i) \right) {,} \label{eq:matrixBi}\\
\bm{\calR}_u(i) &= \diag{(\bmu_k(i)\bmu^\T_k(i))^N_{k=1}} {,} \label{eq:calRui}\\
\calM &= \diag{(\mu_k I_M)^N_{k=1}} {,} \label{eq:calM} \\
\bms(i) &= \calA^\T  \col{(\bmu_k(i) \bmv_k(i))^N_{k=1}}.
\end{align}

\subsection{Mean Error Analysis}
Suppose Assumption \ref{asmp:regressor} and Assumption \ref{asmp:noise} hold, then by taking expectation on both sides of \eqref{eq:error_network} we have the following recursion model for the network mean error,
\begin{align}
\E[\err{\bmw}(i)] = \calB\E[\err{\bmw}(i-1)]  + \calC^\T\E[\bepsilon(i)], \label{eq:error_mean}
\end{align}
where
\begin{align}
\calB &= \E[\bm{\calB}] =\calA^\T\left(I_{MN}-\calM\calR_u\right), \label{eq:matrixB}\\
\calR_u &= \E\left[\bm{\calR}_u(i)\right]=\diag{(R_{u,k})^N_{k=1}}. \label{eq:calRu}
\end{align}

We have the following result on the asymptotic behavior of the mean error.

\begin{Theorem}	\label{thm:mean}
{\rm (Mean Error Stability)} Suppose that Assumption \ref{asmp:regressor} and Assumption \ref{asmp:noise} hold. Then, the network mean error vector of EB-ATC, i.e., $\E[\err{\bmw}(i)]$, is bounded input bounded output (BIBO) stable in steady state if the step-size $\mu_k$ is chosen such that 
\begin{align}
\mu_k < \frac{2}{\lambda_{\max} (R_{u,k})} \label{eq:mean_mu}.
\end{align}
In addition, the block maximum norm of the network mean error is upper-bounded by
\begin{align}
\frac{\alpha}{1-\beta}\cdot\max_{1\le k\le N} \left(\frac{\delta_k}{\lambda_{\min}(Y_k)}\right)^{\frac{1}{2}}, \label{eq:bound_mean_error}
\end{align}
where,
\begin{align}
\alpha=\max_{1\le k\le N}(1-a_{kk}), \;\beta=\norm{I_{MN}-\calM\calR_u}_{b,\infty}.
\end{align}
\end{Theorem} 
\begin{proof}
See Appendix~\ref{sec:appdx_B}
\end{proof}

\subsection{Mean-square Error Analysis}
Due to the triggering mechanism and resulting \emph{a posterior} gap, \eqref{eq:gap_ap_network} correlates with the error vectors \eqref{eq:error_phi} and \eqref{eq:error_w}, and explicitly characterizing the exact network MSD of EB-ATC is technically difficult. Instead, we study the upper bound of the network MSD. First, we derive the MSD recursions as follows. From the recursion \eqref{eq:error_network}, we have the following for any compatible non-negative
definite matrix $\Sigma$:
\begin{align}
\norm{\err{\bmw}(i)}^2_\Sigma = 
&\err{\bmw}(i-1)^\T\bm{\calB}(i)^\T\Sigma\bm{\calB}(i)\err{\bmw}(i-1) + 
\bms(i)^\T\calM^\T\calA\Sigma\calA^\T\calM\bms(i) +
\bepsilon(i)^\T\calC\Sigma\calC^\T\bepsilon(i) \nonumber \\ 
&+2\err{\bmw}(i-1)^\T\bm{\calB}(i)^T\Sigma\calC^\T\bepsilon(i) - 
2\bms(i)\calM^\T\calA\Sigma\calC^\T\bepsilon(i) -
2\err{\bmw}(i-1)^\T\bm{\calB}(i)^\T\Sigma\calA^\T\calM\bms(i) {.} \label{eq:se_recursion}
\end{align}

Taking expectation on both sides of the above expression, the last term evaluates to zero under Assumption~\ref{asmp:regressor}-\ref{asmp:noise}, and we have
\begin{align}
\E\norm{\err{\bmw}(i)}^2_\Sigma = 
&\E\norm{\err{\bmw}(i-1)}^2_{\Sigma'}  + 
t_2 + t_3 + 2t_4 - 2t_5 , \label{eq:mse_recursion}
\end{align}
where the weighting matrix $\Sigma'$ is
\begin{align}
\Sigma'= \E \left[\bm{\calB}(i)^\T\Sigma\bm{\calB}(i)\right] {,} \label{eq:Sigmaprime}
\end{align}
and the last four terms in \eqref{eq:mse_recursion} are given as follows,
\begin{align}
t_2 &= \E[\bms(i)^\T\calM\calA\Sigma\calA^\T\calM\bms(i)], \label{eq:mse_term2} \\
t_3 &= \E[\bepsilon(i)^\T\calC\Sigma\calC^\T\bepsilon(i)], \label{eq:mse_term3} \\
t_4 &= \E[\err{\bmw}(i-1)^\T\bm{\calB}(i)^\T\Sigma\calC^\T\bepsilon(i)], \label{eq:mse_term4}\\
t_5 &= \E[\bms(i)\calM^\T\calA\Sigma\calC^\T\bepsilon(i)]. \label{eq:mse_term5}
\end{align}
Further, we let $\sigma = \Vc(\Sigma)$ and $\sigma' = \Vc(\Sigma')$. We then have $\sigma' = \calE \sigma$, where
\begin{align} 
\calE &= \E\left[ \bm{\calB}(i)^\T \otimes \bm{\calB}(i)^\T \right] \nonumber \\
&= [ I_{M^2N^2} - I_{MN}\otimes\calM\calR_u - \calM\calR_u\otimes I_{MN} + \left( \calM\otimes\calM\right)  \E\left(\bm{\calR}_u(i)\otimes \bm{\calR}_u(i)\right) ] \; \calA\otimes\calA . \label{eq:matrixE}
\end{align}
So that \eqref{eq:mse_recursion} can be rewritten as,
\begin{align}
\E\norm{\err{\bmw}(i)}^2_\sigma = \E\norm{\err{\bmw}(i)}^2_{\calE\sigma} + t_2 + t_3 + 2t_4 - 2t_5 {.} \label{eq:mse_recursion_v2}
\end{align}
Next, we derive the expression and bounds for terms
\subsubsection{Term $t_2$}
For the term $t_2$, we have
\begin{align}
t_2 
&= \E\left[\Tr\left(\calA^\T\calM\bms(i)\bms(i)^\T\calM\calA\Sigma\right)\right] \nonumber\\
&= \Tr\left[\calA^\T\calM\E\left(\bms(i)\bms(i)^\T\right)\calM\calA\Sigma\right] \nonumber\\
&=\Tr\left(\calA^\T\calM\calS\calM\calA \Sigma\right)  \nonumber\\
&= \Vc\left(\calA^\T\calM\calS\calM\calA\right)^\T \sigma , \label{eq:mse_term2_v2} 
\end{align} 
where the equality \eqref{eq:mse_term2_v2} follows from the identity $\Tr(AB) = \Vc(A^T)^T\Vc(B)$, and
\begin{align}
\calS = \diag{(\sigma^2_{v,k} R_{u,k})^N_{k=1}} . \label{eq:matrixS}
\end{align}
\subsubsection{Term $t_3$}
Similarly, we have the following for the term $t_3$,
\begin{align}
t_3 
&= \Tr \left[\calC^\T\E\left(\bepsilon(i)\bepsilon(i)^\T\right)\calC\Sigma\right]  \nonumber\\
&= \Vc(\calC)^\T\left[\Sigma\otimes\E\left(\bepsilon(i)\bepsilon(i)^\T\right) \right]\Vc(\calC)
\end{align}
Moreover, it can be verified that relationship $yy^\T\le y^\T y I_N$ holds for any vector $y\in\Real^N$, and thus  $\bepsilon(i)\bepsilon(i)^\T\le\bepsilon(i)^\T\bepsilon(i) I_{MN}$ follows immediately, so that we have
\begin{align}
\E\left(\bepsilon(i)\bepsilon(i)^\T\right) 
&\le\bepsilon(i)^\T\bepsilon(i) I_{MN} \nonumber\\
&=\sum_{k=1}^{N}\norm{\bepsilon_k(i)}^2 I_{MN} \nonumber\\
&=\sum_{k=1}^{N}\left(\frac{\delta_k}{\lambda_{\min}(Y_k)}\right)^{\frac{1}{2}} I_{MN} {.}
\end{align}
Now, letting 
\begin{align}
\Delta = \sum_{k=1}^{N}\left(\frac{\delta_k}{\lambda_{\min}(Y_k)}\right)^{\frac{1}{2}} ,
\end{align}
due to $\Sigma\ge 0$ the following results follows,
\begin{align}
\Sigma\otimes\left[\E\left(\bepsilon(i)\bepsilon(i)^\T\right) - \Delta I_{MN}\right] \le 0,
\end{align}
and therefore,
\begin{align}
\Vc(\calC)^\T\left\lbrace\Sigma\otimes\left[\E\left(\bepsilon(i)\bepsilon(i)^\T\right) - \Delta I_{MN}\right]\right\rbrace \Vc(\calC) \le 0,
\end{align}
or equivalently,
\begin{align}
\Vc(\calC)^\T\left[\Sigma\otimes\E\left(\bepsilon(i)\bepsilon(i)^\T\right) \right]\Vc(\calC) 
&\le \Delta\cdot\Vc(\calC)^\T\left(\Sigma\otimes I_{MN}\right)\Vc(\calC) \nonumber\\
&=\Delta\cdot\Tr\left(\calC^\T\calC\Sigma\right),
\end{align}
which further implies that
\begin{align}
t_3 \le \Delta\cdot\Vc\left(\calC^\T\calC\right)\sigma {.} \label{eq:mse_term3_bound}
\end{align}
\subsubsection{Term $t_4$}
Since matrix $\Sigma$ is positive semi-definite, so that we have $\Sigma=\Theta\Theta^\T$. Then, let
\begin{align}
P &= \err{\bmw}(i)^\T\bm{\calB}^\T(i)\Theta, \nonumber\\
Q &= \bepsilon(i)^\T\calC\Theta . \label{eq:mse_term4_PQ}
\end{align}
From the fact $(P-Q)(P-Q)^\T\ge0$ we have the following,
\begin{align}
PQ^\T+QP^\T\le PP^\T+QQ^\T.
\end{align}
Substituting \eqref{eq:mse_term4_PQ} into the above inequality and taking expectation on both sides gives,
\begin{align}
2t_4
&\le \E\left[\err{\bmw}(i-1)^\T\bm{\calB}(i)^\T\Sigma\bm{\calB}(i)\err{\bmw}(i-1)\right]  + 
\E\left[\bepsilon(i)^\T\calC\Sigma\calC^\T\bepsilon(i)\right] \nonumber\\
&= \E\norm{\err{\bmw}(i-1)}^2_{\Sigma'} + t_3. \label{eq:mse_term4_bound} 
\end{align}

\subsubsection{Term $t_5$}
Applying manipulations similar with $t_3$ to $t_3$, we have
\begin{align}
t_5 
&= \Tr\left[\calC^\T\E\left(\bepsilon(i)\bms(i)^\T\right)\calM\calA\Sigma\right]  \nonumber\\
&= \Vc\left( \calC^\T\E\left(\bepsilon(i)\bms(i)^\T\right)\calM\calA\Sigma\right) ^\T\sigma . \label{eq:mse_term5_v1}
\end{align}
To facilitate the evaluation of the covariance matrix $\E\left(\bepsilon(i)\bms(i)^\T\right)$, we derive its $(k,\ell)$-th block entry, i.e., $\E\left[\bepsilon_k(i)\bmu_\ell(i)\bmv_\ell(i)\right]$. To this end, substituting \eqref{eq:datamodel} into \eqref{eq:local_update}, we can express $\bpsi_k(i)$ as follows,
\begin{align}
\bpsi_k(i)=\bmw_k(i-1) + \mu_k\bmu_k(i)\bmu_k(i)^\T\err{\bmw}_k(i-1) + \mu_k\bmu_k(i)\bmv_k(i),
\end{align}
so that we have
\begin{align}
\E\left[\bpsi_k(i)\bmu_\ell(i)\bmv_\ell(i)\right] =&
\E\left[\bmw_k(i-1)\bmu_\ell(i)\bmv_\ell(i)\right] + \mu_k\E\left[\bmu_k(i)\bmu_k(i)^\T\err{\bmw}_k(i-1)\bmu_\ell(i)\bmv_\ell(i)\right] \nonumber\\
&+\mu_k\E\left[\bmu_k(i)\bmv_k(i)\bmu_\ell(i)\bmv_\ell(i)\right] . \label{eq:mse_term5_eqn1}
\end{align}
Note that \eqref{eq:mse_term5_eqn1} evaluates to zero if $\ell\neq k$, and when $\ell=k$ the first two terms in \eqref{eq:mse_term5_eqn1} evaluate to zero, and the last term equals $\mu_k\sigma^2_{v,k}R_{u,k}$. In addition, $\E\left[\overline{\bpsi}_k(i)\bmu_\ell(i)\bmv_\ell(i)\right]=0$ for all $\{k,\ell\}\in V$. Therefore, at particular time instant $i$, by conditioning on $\bgamma_k(i)=\gamma_k(i)$ for all $k$, from \eqref{eq:gap_aprior} and \eqref{eq:aposteriorgap_cases} we conclude that
\begin{align}
\E\left[\bepsilon_k(i)\bmu_\ell(i)\bmv_\ell(i)\right] =
\begin{cases}
0, &\text{ if }  \ell\neq k,\nonumber\\
\mu_k\sigma^2_{v,k}R_{u,k}, &\text{ if } \ell=k \text{ and } \gamma_k(i)=0.
\end{cases}
\end{align}
So that the term $t_5$ can be expressed as,
\begin{align}
t_5 = -\Vc\left(\calC^\T\calG(i)\calM\calS\calM\calA\right)^\T\sigma, \label{eq:mse_term5_v2}
\end{align}
where matrix $\calS$ is given in \eqref{eq:matrixS} and
\begin{align}
\calG(i) = \E\diag{(\bgamma_k(i) I_M)^N_{k=1}}-I_{MN} . \label{eq:matrixG}
\end{align}
Therefore, substituting \eqref{eq:mse_term2_v2}, \eqref{eq:mse_term3_bound}, \eqref{eq:mse_term4_bound}, and \eqref{eq:matrixG} into \eqref{eq:mse_recursion_v2}, we have the following bound for the network MSD at time instant $i$,
\begin{align}
\E\norm{\err{\bmw}(i)}^2_\sigma\le
&\E\norm{\err{\bmw}(i-1)}^2_{\calD\sigma} + [f_1+f_2+f_3(i)]^\T\sigma, \label{eq:msd_bound_recursion}
\end{align}
where $\calD = 2 \calE$ and matrix $\calE$ is given in \eqref{eq:matrixE}, and
\begin{align}
f_1&=\Vc\left(\calA^\T\calM\calS\calM\calA\right), \nonumber\\
f_2&=2\Delta\cdot\Vc\left(\calC^\T\calC\right), \nonumber\\
f_3(i)&=2\Vc\left( \calC^\T\calG(i)\calM\calS\calM\calA\right) . \label{eq:f_123}
\end{align}

\begin{Assumption}\label{asmp:Ruk}
	Each node $k$ adopts a regressor covariance matrix $R_{u,k}$ whose eigenvalues satisfy
		\begin{align}
		\lambda_\max(R_{u,k}) < \left(\frac{2+\sqrt{2}}{2-\sqrt{2}}\right) \lambda_\min(R_{u,k}). \label{eq:assumpRuk}
		\end{align}
\end{Assumption}
\begin{Theorem}\label{thm:msd}
{\rm (Mean-square Error Behavior)} Suppose that Assumptions \ref{asmp:regressor}-\ref{asmp:noise} hold. Then, as $i\rightarrow\infty$, the network MSD of EB-ATC, i.e.,  $\E\norm{\err{\bmw}(i)}^2/N$, has a finite constant upper bound if the step sizes $\{\mu_k\}$ are chosen such that $\rho(\calD)<1$ is satisfied. In addition, it follows that matrix $\calD$ can be approximated by $\calD\approx\calF=\calD+O(\calM^2)$, where
\begin{align}
\calF= 2\calB^\T\otimes\calB^\T, \label{eq:matrixF} 
\end{align}
so that if Assumption~\ref{asmp:Ruk} also holds and $\{\mu_k\}$ also satisfy 
\begin{align}
\frac{1-\frac{\sqrt{2}}{2}}{\lambda_{\min}(R_{u,k})}<\mu_k < \frac{1+\frac{\sqrt{2}}{2}}{\lambda_{\max} (R_{u,k})}, \label{eq:msd_bound_mu}
\end{align}
an upper bound of the network MSD in steady state is given by
\begin{align}
\frac{1}{N}\left[(f_1+f_2)^\T\left(I_{M^2N^2}-\calF\right)^{-1} + f_{3,\infty}\right]\Vc(I_{MN}) +O(\mu_{\operatorname{max}}^2),
\end{align}
where,
\begin{align}
\mu_{\operatorname{max}}&=\max_{1\le k\le N}\{\mu_k\}, \label{eq:max_stepsize}\\
f_{3,\infty}&= \lim\limits_{i\rightarrow\infty} \sum_{j=0}^{i-1}f_3(i-j)^\T\calF^{j}. \label{eq:f3_infty}
\end{align}
\end{Theorem}
\begin{proof}
See Appendix~\ref{sec:appdx_C}
\end{proof}

\begin{Remark}
Assumption~\ref{asmp:Ruk} is needed additionally to ensure that the set of $\mu_k$ in \eqref{eq:msd_bound_mu} is non-empty. Note that if $R_{u,k}$ is chosen to be $R_{u,k} = \sigma^2_{u,k} I_M$, the above assumption \eqref{eq:assumpRuk} is automatically met, and condition \eqref{eq:msd_bound_mu} becomes
\begin{align}
\frac{2-\sqrt{2}}{2\sigma^2_{u,k}}<\mu_k < \frac{2+\sqrt{2}}{2\sigma^2_{u,k}}. \label{eq:msd_bound_mu2}
\end{align}
Besides, although diffusion adaptation strategies \cite{CatSayed:J10,ZhaoSayed:J12,TuSayed:J12,SayedTuChen:M13,Sayed:BC14,Sayed:IP14} usually do not have lower bounds for step sizes on the stability of network MSD, the condition \eqref{eq:msd_bound_mu} is a sufficient condition to ensure the upper bound of the network MSD \eqref{eq:msd_bound_recursion} converges at steady state, so that \eqref{eq:msd_bound_mu} is only sufficient (but not necessary) for the stability of the exact network MSD in steady state. Indeed, numerical studies also suggest that without relying on Assumption~\ref{asmp:Ruk} and choosing a step size even smaller than the lower bounds in \eqref{eq:msd_bound_mu} will not cause the divergence of the network MSD in steady state.
\end{Remark}

\section{Simulation Results}\label{sec:simulation}
In this section, numerical examples are provided to illustrate the MSD performance and energy-efficiency of the proposed EB-ATC, and to compare against ATC and the non-cooperative LMS algorithm. We performed simulations on a network with $N=60$ nodes as depicted in Fig.~\ref{fig:sim}(a). The measurement noise powers $\{\sigma^2_{v,k}\}$ are generated from a uniform distribution over $[-25, -10]$ dB. We consider a parameter of interest $w^\circ$ with dimension $M=10$, and suppose that the zero-mean regressor $u_k(i)$ has covariance $R_{u,k} = \sigma^2_{u,k} I_M$, where the coefficients $\{\sigma^2_{u,k}\}$ are drawn uniformly from the interval $[1,2]$. For the ease of implementation, we adopt constant and uniform triggering thresholds $\delta_k(i)=\delta$, and identity weighting matrix $Y_k=I_M$ for the event triggering function of every node. Moreover, we use the \emph{Metropolis} rule \cite{Sayed:BC14} for the diffusion combination \eqref{eq:eb_atc_combination}. All the simulations results are averaged over 200 Monte Carlo runs.

\begin{figure*}[!t]
	\centering
	\subfloat[Network topology]{\includegraphics[width=0.28\textwidth]{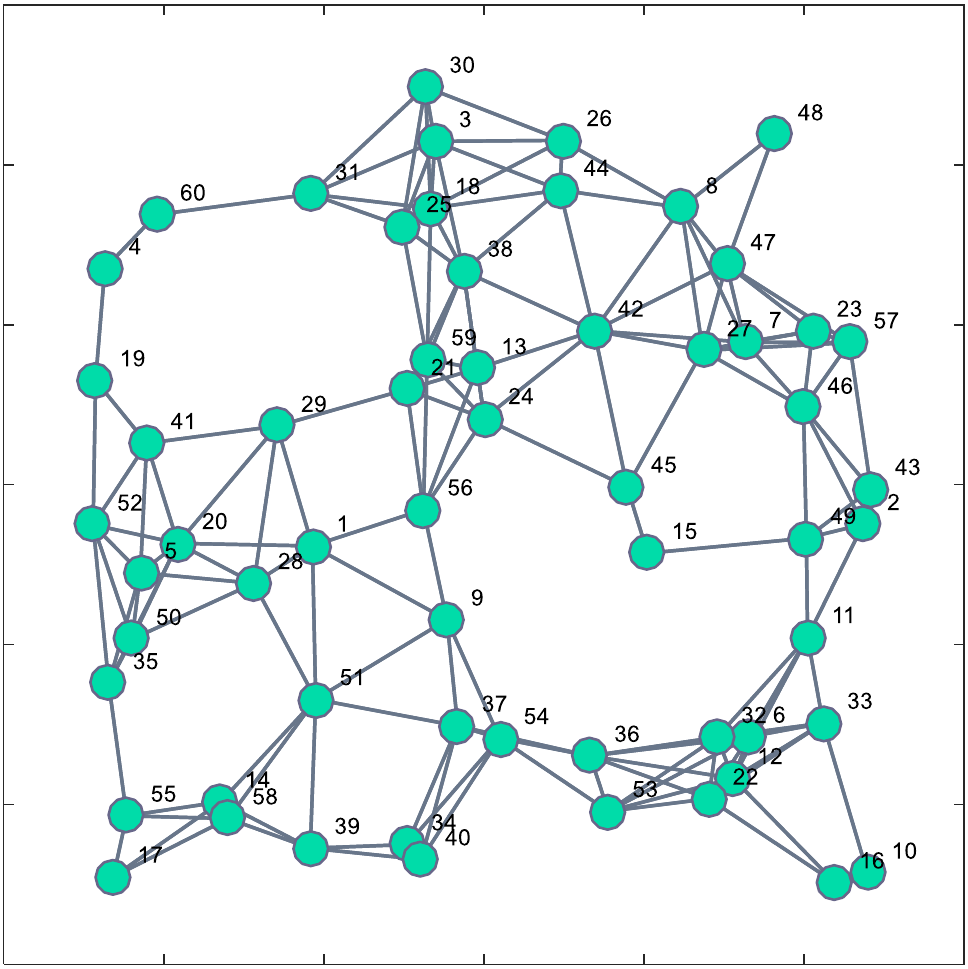}%
		\label{fig:topology}}
	\hfil
	\subfloat[MSD performance]{\includegraphics[width=0.35\textwidth]{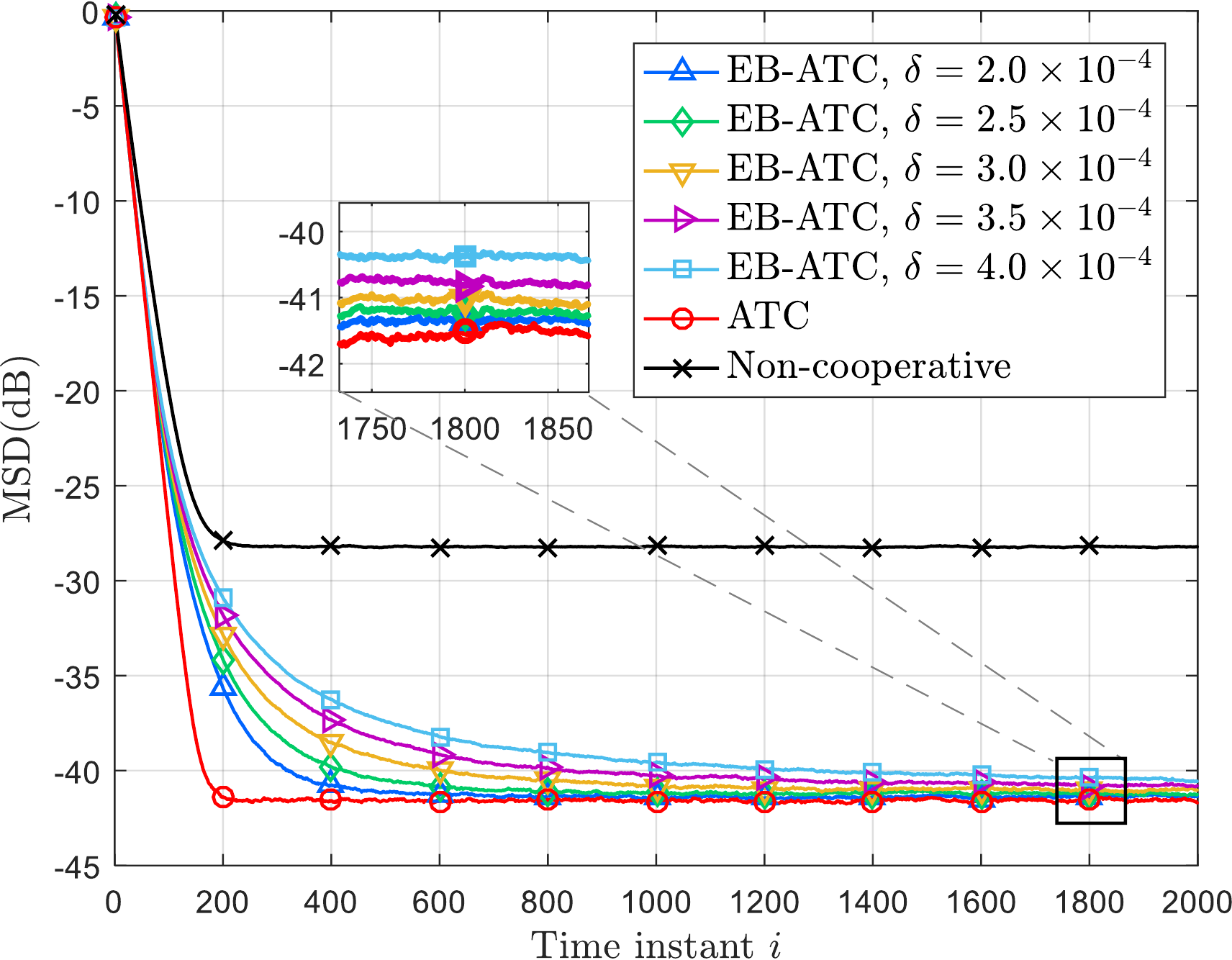}%
		\label{fig:msd}}
	\hfil
	\subfloat[Average ENTR]{\includegraphics[width=0.35\textwidth]{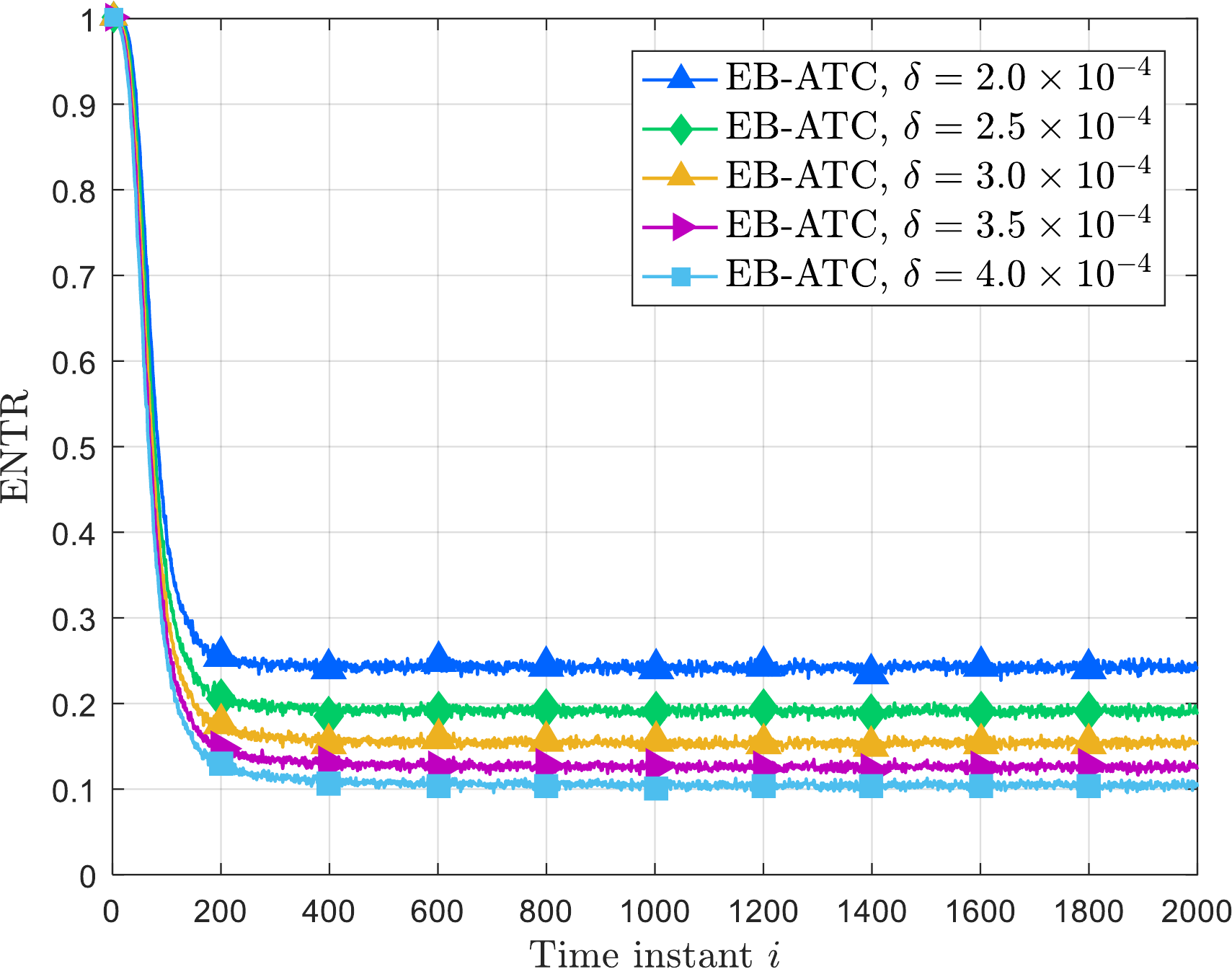}%
		\label{fig:entr}}	
	\caption{Simulation results for the network.}
	\label{fig:sim}
\end{figure*}

From Fig.~\ref{fig:sim}(b), it can be observed that compared with the ATC strategy, MSDs of the proposed EB-ATC in steady-state are higher by a few dBs, but still much lower than that of the non-cooperative LMS algorithm, which demonstrates the capability of EB-ATC to preserve the benefits of diffusion cooperation. On the other hand, the convergence of EB-ATC is relatively slower. This is because in the transient phase, the event-based communication mechanism of EB-ATC restricts the frequency of exchanging the newest local intermediate estimates $\{\bpsi_{k,i}\}$, for the purpose of energy saving. This leads to inferior transient performance compared to ATC. 

On the other hand, EB-ATC achieves significant communication overhead savings compared to ATC. To visualize this, we define the expected network triggering rate (ENTR) as follows:
\begin{align}
\operatorname{ENTR}(i) = \frac{1}{N}\sum_{k=1}^{N} \E \gamma_k(i). \label{eq:ENTR}
\end{align}
The ENTR at time instant $i$ captures how frequently communication is triggered by each node at that time instant $i$, on average. ENTR is directly proportional to the average communication overhead incurred by the nodes in the network at each time instant. From \eqref{eq:ENTR}, it is clear that $0\le\operatorname{ENTR}(i)\le1$, so a smaller value of $\operatorname{ENTR}(i)$ implies a lower energy consumption. Note that ATC has ENTR$(i)=1$ for all time instants $i$. From Fig.~\ref{fig:sim}(c), we observe that the ENTR for EB-ATC decays rapidly over time during the transient phase, and for all the different triggering thresholds we tested, EB-ATC uses less than 30\% of the communication overhead of ATC after the time instant $i\approx 200$, which is the average time that the MSD of ATC is within 90\% of its steady-state value. This demonstrates that even though EB-ATC has not reached steady-state (at $i\approx 600$), communication between nodes do not trigger very frequently as the intermediate estimates do not change significantly after this time instant. Furthermore, in steady state, although each node maintains estimates that are close to the true parameter value, communication triggering does not completely stop. This is due to occasional abrupt changes in the random noise and regressors, which can make the local estimate update deviate significantly. This is in the same spirit of why MSD does not converge to zero. 

It is also worth mentioning that, although in theory the methods in the literature \cite{ArabloueiLMS:J14,SayinKozat:J14,HarraneRicard:C16} can save more energy by transmitting only a few entries or compressed values, for real-time applications they may not be as reliable as EB-ATC in under the same channel conditions, especially when the SNR is poor. To guarantee successful diffusion cooperation among neighborhood, higher channel SNR or more robust encoding scheme is required for \cite{ArabloueiLMS:J14,SayinKozat:J14,HarraneRicard:C16}, whereas EB-ATC is simpler yet effective.

\section{Conclusion}\label{sec:conclusion}
We have proposed an event-based diffusion ATC strategy where communication among neighboring nodes is triggered only when significant changes occur in the local updates. The proposed algorithm is not only able to significantly reduce communication overhead, but can still maintain good MSD performance at steady-state compared with the conventional diffusion ATC strategy. Future research includes analyzing the expected triggering rate theoretically as well as characterizing the rate of convergence, and to establish their relationship with the triggering threshold, so that the thresholds can be selected to optimize its performance.

\appendices
\section{Proof of Lemma~\ref{lemma:lemma1}}\label{sec:appdx_A}
Since $Y_k$ is positive semi-definite, and therefore real symmetric, so that there exists an unitary matrix $U$ such that
\begin{align}
Y_k = U \diag{\lambda_m(Y_k)_{m=1}^N} U^\T,
\end{align}
Let $\phi_m, m=1,2,\dots,M$ be the eigenvectors of $Y_k$, so we have
\begin{align}
U = [\phi_1, \phi_2, \cdots, \phi_M].
\end{align}
Recall that any vector $x\in\Real^M$ can be expressed as
\begin{align}
x = \sum_{m}(\phi_m^\T x) \phi_m,
\end{align}
therefore it is easy to verify that
\begin{align}
\frac{\norm{x}^2_{Y_k}}{\norm{x}^2}=\frac{x^\T Y_k x}{x^\T x}=\frac{\sum_{m}\lambda_m(Y_k)(\phi_m^\T x)^2}{\sum_{m}(\phi_m^\T x)^2}\ge\lambda_{\min}(Y_k)
\end{align}
which implies
\begin{align}
\lambda_{\min}(Y_k)\cdot\norm{\bepsilon_k(i)}^2\le\norm{\bepsilon_k(i)}^2_{Y_k}.
\end{align}
Besides, from \eqref{eq:aposteriorgap_cases}, we can conclude that
\begin{align}
\norm{\bepsilon_k(i)}^2_{Y_k}\le\norm{\bepsilon_k^-(i)}^2_{Y_k}\le\delta_k(i)
\end{align} 
Therefore, we have 
\begin{align}
\lambda_{\min}(Y_k)\cdot\norm{\bepsilon_k(i)}^2\le\delta_k(i)\le\delta_k
\end{align}
which gives 
\begin{align}
\norm{\bepsilon_k(i)}\le\left(\frac{\delta_k}{\lambda_{\min}(Y_k)}\right)^{\frac{1}{2}}. 
\end{align}
The proof is complete.

\section{Proof of the Theorem~\ref{thm:mean}}\label{sec:appdx_B}
Taking block maximum norm $\norm{\cdot}_{b,\infty}$ to $\E[\bepsilon(i)]$, due to every norm is a convex function of its argument, by Jensen's inequality and Lemma~\ref{lemma:lemma1}, we have
\begin{align}
\norm{\E[\bepsilon(i)]}_{b,\infty} 
&\le \E\left[\norm{\bepsilon(i)}_{b,\infty}\right] \\
&= \E \left[\max_{1\le k\le N}\norm{\bepsilon_k(i)}\right] \label{eq:bound_apgap_equality} \\
&\le \max_{1\le k\le N} \left(\frac{\delta_k}{\lambda_{\min}(Y)}\right)^{\frac{1}{2}}, \label{eq:bound_apgap}
\end{align}
where we have used the definition of the block maximum norm in \cite{Sayed:BC14} for the equality \eqref{eq:bound_apgap_equality}, and \eqref{eq:bound_apgap} follows from the Lemma~\ref{lemma:lemma1}. The right hand side (R.H.S.) of \eqref{eq:bound_apgap} is a finite constant scalar, which implies that the input signal to the recursion \eqref{eq:error_mean}, i.e., $\E[\bepsilon(i)]$ is bounded. Therefore, the recursion \eqref{eq:error_mean} is BIBO stable if $\rho(\calB)<1$. 

In addition, since matrix $\calA^\T$ is left-stochastic, by applying the Lemma.~D5 and Lemma.~D6 in \cite{Sayed:BC14}, we have the following from \eqref{eq:matrixB}, 
\begin{align}
\rho(\calB)
&=\rho\left(\calA^\T\left(I_{MN}-\calM\calR_u\right)\right) \label{eq:rho_calB}\\
&\le\rho\left(I_{MN}-\calM\calR_u\right)\\
&=\norm{I_{MN}-\calM\calR_u}_{b,\infty}.
\end{align}
Therefore, we conclude that the network mean error is BIBO stable if 
\begin{align}
\norm{I_{MN}-\calM\calR_u}_{b,\infty} < 1, \label{eq:rho_calB_2}
\end{align}
which further yields the condition \eqref{eq:mean_mu}.

To establish the upper bound \eqref{eq:bound_mean_error}, we iterate \eqref{eq:error_mean} from $i=0$, which gives,
\begin{align}
\E[\err{\bmw}(i)]=\calB^i \E[\err{\bmw}(0)] + \sum_{j=0}^{i-1}\calB^{j} \calC^\T \E[\bepsilon(i-j)].
\end{align}
Then applying block maximum norm $\norm{\cdot}_{b,\infty}$ on both sides of the above equation, by the properties of vector norms and induced matrix norms, it can be obtained that
\begin{align}
\norm{\E[\err{\bmw}(i)]}_{b,\infty} 
&\le \norm{\calB^i}_{b,\infty}\cdot\norm{\E[\err{\bmw}(0)]}_{b,\infty} + \sum_{j=0}^{i-1}\norm{\calB^{j}}_{b,\infty}\cdot\norm{\calC^\T \E[\bepsilon(i-j)]}_{b,\infty} \\
&\le \norm{\calA^\T}_{b,\infty}^i\cdot\norm{I_{MN}-\calM\calR_u}_{b,\infty}^i\cdot\norm{\E[\err{\bmw}(0)]}_{b,\infty}\\ &\qquad+\sum_{j=0}^{i-1}\norm{\calA^\T}_{b,\infty}^{j}\cdot\norm{I_{MN}-\calM\calR_u}_{b,\infty}^{j}\cdot\norm{\calC^\T}_{b,\infty}\cdot\norm{\E[\bepsilon(i-j)]}_{b,\infty} .
\label{eq:bound_meanerror_1}
\end{align}
Let $\alpha=\norm{\calC^\T}_{b,\infty}$, from the Lemma.~D3 of \cite{Sayed:BC14} we have 
\begin{align}
\alpha=\norm{C^\T}_\infty=\max_{1\le k\le N}(1-a_{kk}).
\end{align}
Moreover, since matrix $\calA^\T$ is left-stochastic, so that we have $\norm{\calA^\T}_{b,\infty}=1$ by the Lemma.~D4 of \cite{Sayed:BC14}. Let $\beta=\norm{I_{MN}-\calM\calR_u}_{b,\infty}$, then substitute \eqref{eq:bound_apgap} into \eqref{eq:bound_meanerror_1} we obtain that,
\begin{align}
\norm{\E[\err{\bmw}(i)]}_{b,\infty} \le \norm{\E[\err{\bmw}(0)]}_{b,\infty}\cdot\beta^i + \alpha\cdot\max_{1\le k\le N} \left(\frac{\delta_k}{\lambda_{\min}(Y)}\right)^{\frac{1}{2}}\cdot\sum_{j=0}^{i-1}\beta^j . \label{eq:bound_mean_error_2}
\end{align}
If step size $\mu_k$ is chosen to satisfy $0\le\beta<1$, then letting $i\rightarrow\infty$ on both sides of \eqref{eq:bound_mean_error_2} we arrive at following inequality relationship
\begin{align}
\lim_{i\rightarrow\infty}\norm{\E[\err{\bmw}(i)]}_{b,\infty} \le \frac{\alpha}{1-\beta}\cdot\max_{1\le k\le N} \left(\frac{\delta_k}{\lambda_{\min}(Y)}\right)^{\frac{1}{2}}, \label{eq:bound_meanerror_3} 
\end{align}
and the proof is complete.	

\section{Proof of the Theorem~\ref{thm:msd}}\label{sec:appdx_C}
To obtain the upper bound of network MSD at steady state, iterating \eqref{eq:msd_bound_recursion} from $i=1$, we have 
\begin{align}
\E\norm{\err{\bmw}(i)}^2_\sigma\le
\E\norm{\err{\bmw}(i-1)}^2_{\calD\sigma} + 
(f_1+f_2)^\T\left(\sum_{j=0}^{i-1}\calD^{j}\right)\sigma+
\left(\sum_{j=0}^{i-1}f_3(i-j)^\T\calD^{j}\right)\sigma, \label{eq:msd_bound_recursion_v2}
\end{align}
where vectors $f_1$, $f_2$, and $f_3(i)$ are given in \eqref{eq:f_123}. Letting $i\to\infty$, the first term on the R.H.S. of the above inequality converges to zero, and the second term converge to a finite value $(f_1+f_2)^\T\left(I_{M^2N^2}-\calD\right)^{-1}\sigma$, if and only if $\calD^i \rightarrow 0$ as $i \rightarrow \infty$, i.e., $\rho(\calD)<1$. From \eqref{eq:matrixG} and \eqref{eq:f_123} we have $f_3(i)$ is bounded due to every entry of matrix $\calG(i)$ is bounded. Moreover, if $\rho(\calD)<1$,  there exists a norm $\norm{\cdot}_\zeta$ such that $\norm{\calD}_\zeta<1$, therefore we have 
\begin{align}
\left|f_3(i-j)^\T\calD^j\sigma\right|\le a \cdot\norm{\calD}^j_\zeta,
\end{align}
for some positive constant $a$. Since $\norm{\calD}_\zeta^j\rightarrow0$ as $j\rightarrow\infty$, the series,
\begin{align}
\sum_{j=0}^{i-1} \left| f_3(i-j)^\T\calD^{j}\sigma\right| 
\end{align}
converges as $i\rightarrow\infty$, which implies the absolute convergence of the third term of R.H.S of \eqref{eq:msd_bound_recursion_v2}.

Besides, note that the matrix $\calF$ given in \eqref{eq:matrixF} can be explicitly expressed as
\begin{align} 
\calF &= 2 \calB^\T \otimes \calB^\T  \nonumber \\
&= [ I_{M^2N^2} - I_{MN}\otimes\calM\calR_u - \calM\calR_u\otimes I_{MN} 
+ \left( \calM\otimes\calM\right)  \left(\calR_u\otimes \calR_u\right) ] \; \calA\otimes\calA . \label{eq:matrixFF}
\end{align}
Substituting \eqref{eq:matrixE} in to $\calD=2\calE$ and comparing with the above \eqref{eq:matrixFF}, we have
\begin{align}
\calD=\calF+O(\calM^2), \label{eq:calDcalF}
\end{align}
where
\begin{align}
O(\calM^2)=\left( \calM\otimes\calM\right) \left\lbrace\E\left[\bm{\calR}_u(i)\otimes \bm{\calR}_u(i)\right]-\calR_u\otimes\calR_u\right\rbrace, \label{eq:OM2}
\end{align}
so that substituting \eqref{eq:calDcalF} into the R.H.S of \eqref{eq:msd_bound_recursion_v2} gives
\begin{align}
\E\norm{\err{\bmw}(i)}^2_\sigma\le
\E\norm{\err{\bmw}(i-1)}^2_{\calF\sigma} + 
(f_1+f_2)^\T\left(\sum_{j=0}^{i-1}\calF^{j}\right)&\sigma+
\left(\sum_{j=0}^{i-1}f_3(i-j)^\T\calF^{j}\right)\sigma \nonumber\\
&+\E\norm{\err{\bmw}(i-1)}^2_{O(\calM^2)\sigma} + g(i)^\T O(\calM^2)\sigma, \label{eq:msd_bound_recursion_v3}
\end{align}
where
\begin{align}
g(i)=f_1+f_2+\sum_{j}^{i-1}f_3(j). \label{eq:vector_gi}
\end{align}
Due to the vector $g(i)$ is bounded, so that if $\rho(\calD)<1$ such that $\E\norm{\err{\bmw}(i-1)}^2_\sigma$ is bounded, then the last two terms on the R.H.S of \eqref{eq:msd_bound_recursion_v3} are negligible for sufficiently small step sizes $\{\mu_k\}$, which means the matrix $\calD$ can be approximated by $\calD\approx\calF$ if $\{\mu_k\}$ are sufficiently small and also satisfy $\rho(\calD)<1$. Therefore \eqref{eq:msd_bound_recursion_v3} can be further expressed as
\begin{align}
\E\norm{\err{\bmw}(i)}^2_\sigma\le
\E\norm{\err{\bmw}(i-1)}^2_{\calF\sigma} + 
(f_1+f_2)^\T\left(\sum_{j=0}^{i-1}\calF^{j}\right)\sigma+
\left(\sum_{j=0}^{i-1}f_3(i-j)^\T\calF^{j}\right)\sigma+
O(\mu_{\operatorname{max}}^2) {.} \label{eq:msd_bound_recursion_v4}
\end{align}
Choosing $\sigma=\frac{\Vc(I_{MN})}{N}$ and using arguments similar for \eqref{eq:msd_bound_recursion_v2}, as $i\to\infty$ the first term on the R.H.S of \eqref{eq:msd_bound_recursion_v4} converges to
\begin{align}
\frac{1}{N}[(f_1+f_2)\left(I_{M^2N^2}-\calF\right)^{-1} + g_\infty]\Vc(I_{MN}) {,}
\end{align}
if and only if $\calF$ is stable, i.e., $\rho(\calF)<1$, where $f_{3,\infty}$ is given in \eqref{eq:f3_infty}. Since
\begin{align}
\rho(\calF)=\rho(2\calB^\T\otimes\calB^\T)=2\rho(\calB)^2,
\end{align}
so that a sufficient condition to guarantee $\rho(\calF)<1$ is $\rho(\calB)<\frac{\sqrt{2}}{2}$. By the Lemma~D.5 in \cite{Sayed:BC14}, we have 
\begin{align}
\rho(\calB)\le \rho(I_{MN}-\calM \calR_{u})=\max_{1\le k\le N} \rho(I_{M}-\mu_k R_{u,k}).
\end{align}
Thus, to have $\rho(\calB)<\frac{\sqrt{2}}{2}$, we need 
\begin{align}
\max_{1\le k\le N} \rho(I_{M}-\mu_k R_{u,k})<\frac{\sqrt{2}}{2},
\end{align}
which is requiring each node $k$ to satisfy
\begin{align}
\max_{1\le m \le M}|1-\mu_k \lambda_m(R_{u,k})|<\frac{\sqrt{2}}{2},
\end{align}
and this is equivalent to require that
\begin{align}
|1-\mu_k \lambda_m(R_{u,k})|<\frac{\sqrt{2}}{2} \label{eq:lambdaRuk}
\end{align}
holds for each eigenvalue of $R_{u,k}$, i.e.,  $\lambda_m(R_{u,k})$. From \eqref{eq:lambdaRuk}, we obtain that $\mu_k$ needs to satisfy
\begin{align}
\frac{1-\frac{\sqrt{2}}{2}}{\lambda_{m}(R_{u,k})}<\mu_k < \frac{1+\frac{\sqrt{2}}{2}}{\lambda_{m} (R_{u,k})} \label{eq:mu_k_lambda_m}
\end{align}
for each of $\{\lambda_m(R_{u,k})|1\le m \le M\}$. In addition, suppose for each $R_{u,k}$ we have
\begin{align}
\lambda_\max(R_{u,k}) < \left(\frac{2+\sqrt{2}}{2-\sqrt{2}}\right) \lambda_\min(R_{u,k}),
\end{align}
then requiring $\mu_k$ to satisfy \eqref{eq:mu_k_lambda_m} for every $\lambda_m(R_{u,k})$ yields 
\begin{align}
\frac{1-\frac{\sqrt{2}}{2}}{\lambda_{\min}(R_{u,k})}<\mu_k < \frac{1+\frac{\sqrt{2}}{2}}{\lambda_{\max} (R_{u,k})}, \nonumber
\end{align}
which is the condition \eqref{eq:msd_bound_mu}.
\bibliographystyle{IEEEtran}

\end{document}